\documentclass[runningheads]{llncs}
\usepackage{amsmath}
\usepackage{algorithm}
\usepackage{algorithmic}
\usepackage{graphicx}
\usepackage{url}
\begin{document}

\title{Improving the complexity of Parys' recursive algorithm\thanks{Supported by EPSRC project Solving Parity Games in Theory and Practice.}}
%
%
\author{K. Lehtinen \and
S. Schewe \and
D. Wojtczak}
%
%
\institute{University of Liverpool, Liverpool, UK\\
\email{[k.lehtinen,sven.schewe,d.wojtczak]@liverpool.ac.uk}\\
}
\maketitle              
%
\newcommand{\solve}{{\textsc{Solve}}}
\newcommand{\atr}{\textsc{Attr}}

\begin{abstract}
Parys has recently proposed a quasi-polynomial version of Zielonka's recursive algorithm for solving parity games.
In this brief note we suggest a variation of his algorithm that improves the complexity to meet the state-of-the-art complexity of broadly $2^{O((\log n)(\log c))}$, while providing polynomial bounds when the number of colours is logarithmic.

\end{abstract}
\section{Introduction}

 In 2017 Calude et al. published the first quasi-polynomial algorithm for solving parity games~\cite{CJKLS17}.
 Since then, several alternative algorithms have appeared~\cite{Leh18,JL17},
 the most recent of which is Parys's quasi-polynomial version of the Zielonka's recursive algorithm~\cite{Par19}.
 
Parys's algorithm, although enjoying much of the conceptual simplicity of Zielonka's algorithm~\cite{Zie98}, has a complexity that is a quasi-polynomial factor larger than~\cite{CJKLS17},~\cite{JL17}, and~\cite{FJSSW17}.
More precisely, their complexity is, modulo a small polynomial factor, ${{c'+l}\choose{l}}$, with $c'$ being $c$ or $c/2$ and $l \in O(\log n)$, for games with $n$ positions and $c$ colours.
This also provides fixed-parameter tractability and  a polynomial bound for the common case where the number of colours is logarithmic in the number of states.
We propose a simplification that brings the complexity of Pary's algorithm down to match this. 
Note, however, that in a fine grained comparison the recursive algorithm still operates symmetrically, going through every colour, rather than just half of them, and $O(\log n)$ hides a factor of $2$.
Thus, a very careful analysis still reveals a small gap.

We also briefly comment on the relationship between this recursive algorithm and universal trees.

\section{Preliminaries}

A parity game $G=(V,V_E,E,\Omega:V\rightarrow [0..c])$ is a two-player game between players Even and Odd, on a finite graph $(V,E)$,
of which positions are partitioned between those belonging to Even, $V_E$ and those belonging of Odd $V_O=V\setminus V_E$, and
labelled by $\pi$ with integer colour from a finite co-domain $[0..c]$ by $\pi$. We assume that every position has a successor.

A play $\pi$ is an infinite path through the game graph.
It is winning for Even if the highest colour occurring infinitely often on it is even; else it is winning for Odd.
We write $\pi[i]$ for the $i^{th}$ position in $\pi$ and $\pi[0,j]$ for its prefix of length $j+1$.

A strategy for a player maps every prefix of a play ending in a position that belongs to this player to one of its successors.
A play $\pi$ agrees with a strategy $\sigma$ for Even (Odd) if whenever $\pi[i]\in V_E$ ($V_O$), then $\sigma(\pi[0,i])=\pi[i+1]$.
A strategy for a player is winning from a position $v$ if all plays beginning at $v$ that it agrees with are winning for that player.
Parity games are determined: from every position, one of the two players has a winning strategy~\cite{Mar75}.

Even's (Odd's) winning region in a parity game is the set of nodes from which Even (Odd) has a winning strategy.
We are interested in the problem of computing, given a parity game $G$, the winning regions of each player.

Given a set $S\subseteq V$, the E-attractor of $S$ in $G$, written $\atr_E(S,G)$,
is the set of nodes from which Even has a strategy which only agrees with plays that reach $S$.
O-attractors, written $\atr_O(S,G)$ are defined similarly for Odd.

An even dominion is a set of nodes $P\subseteq V$ such that nodes in $P\cap V_E$ have at least one successor in $P$ and nodes in $P\cap V_O$ have all of their successors in $P$, and Even has a winning strategy within the game induces by $P$. 
An odd dominion is defined similarly.

We will use the following simple lemmas to prove the correctness of our algorithm.

\begin{lemma}\label{lem:att}
If a dominion $D$ for player $P$ in a game $G$ does not intersect with $X$, then it does not intersect with $\atr_{\bar P}(X,G)$ either, where $\bar P$ is the opponent of $P$.
\end{lemma}

\begin{proof}
From the definition of a dominion, player $P$ has a strategy that from within $D$ only agrees with plays staying within $D$, contradicting any node in $D$ being within the attractor of $X$.
\end{proof}

\begin{lemma}\label{lem1}
Let $D$ be a dominion for player $P$ in a game $G$. Then for all sets $X$, $D\setminus \atr_{P}(X,G)$ is a  dominion for $P$ in $G\setminus  \atr_{P}(X,G)$.
\end{lemma}

\begin{proof}
The same strategy that witnesses $D$ being a dominion for $P$ in $G$ witnesses $D\setminus \atr_{P}(X,G)$ being a  dominion for $P$ in $G\setminus  \atr_{P}(X,G)$.
\end{proof}

\begin{lemma}\label{lem:no-h}
If the highest priority $h$ in a dominion $D$ for a player $P$ in a play $G$ is not of $P$'s parity, then $D$ contains a non-empty sub-dominion without $h$.
\end{lemma}

\begin{proof}
Otherwise, every position in $D$ would be in the attractor of the nodes of priority $h$, and the opponent would have a strategy to see $h$ infinitely often.
\end{proof}

\section{The Algorithm}

\begin{algorithm}
\caption{$\solve_E(G,h,p_E,p_O)$}
\begin{algorithmic}[1]\label{alg:parys}
\IF{$G = \emptyset \vee p_E \leq 1$}
\RETURN $\emptyset;$
\ENDIF
\WHILE{$W_O\neq 0$}
\STATE $N_h := \{v\in G | \pi(v)=h\};$
\STATE $H := G \setminus \atr_E(N_h,G)$
\STATE $W_O := \solve_O(H,h-1,\lfloor p_O/2\rfloor,p_E);$
\STATE $G := G \setminus \atr_O(W_O,G);$
\ENDWHILE
\STATE $N_h := \{v\in G | \pi(v)=h\};$
\STATE $H := G \setminus \atr_E(N_h,G,)$
\STATE $W_O := \solve_O(H,h-1,p_O,p_E);$
\STATE $G := G \setminus \atr_O(W_O,G,);$

\WHILE{$W_O\neq 0$}
\STATE $N_h := \{v\in G | \pi(v)=h\};$
\STATE $H := G \setminus \atr_E(N_h,G)$
\STATE $W_O := \solve_O(H,h-1,\lfloor p_O/2\rfloor,p_E);$
\STATE $G := G \setminus \atr_O(W_O,G);$
\ENDWHILE
\RETURN $G$
\end{algorithmic}
\end{algorithm}
We first recall Parys' quasi-polynomial version of Zielonka's algorithm in Algorithm~\ref{alg:parys}.
In brief, the difference between this algorithm and Zielonka's is that this procedure takes a pair of parameters that bound the size of the dominions, for Eve and Odd respectively, that the procedure looks for;
it first removes one player's dominions (and their attractors) of size up to half the parameter until this does not yield anything anymore,
then searches for a single dominion of the size up to the input parameter, then again carries on with searching for small dominions. 
In each of the recursive calls, the algorithm solves a parity game with one colour less, and either half the input parameter (most of the time)
or the full input parameter (once).
The correctness hinges on the observation that only one dominion can be larger than half the size of the game,
so the costliest call with the full size of the game as parameter needs to be called just once.

Our simplification, in Algorithm~\ref{alg:simple}, replaces each of the two while-loops with a single recursive call that also halves a precision parameter, but, unlike Parys's algorithm, operates on the whole input game arena at once (or what is left of it, in the case of the last call), rather than on a series of subgames of lower maximal colour. In brief, our algorithm computes three regions $W_1$, $W_2$ and $W_3$ one after the other which together contain all small Odd dominions and no small Even dominion, in order to return the complement of their union. $W_1$ and $W_3$ are based on calls that only identify very small dominions; the correctness of the algorithm hinges on proving that $W_1$ and $W_2$ together already account for over half of any small Odd dominion, and hence the last call will correctly handle what is left.

For both algorithms, the dual, $\solve_O$ is defined by replacing $E$ with $O$ and vice-versa, although in our case $\solve_O$ returns $\emptyset$, rather than $G$, if $h=0$.

\begin{algorithm}\label{alg:simple}
\caption{$\solve_E(G,h,p_E,p_O)$}
\begin{algorithmic}[1]
\IF{$G = \emptyset$}
\RETURN $\emptyset;$
\ENDIF
\IF{$p_O = 0$ or $h=0$}
\RETURN $G$
\ENDIF
\STATE $W=G\setminus \solve_E(G,h, p_E, \lfloor p_O/2\rfloor);$\label{W}
\STATE $W_1=\atr_O(W_O',G);$\label{W1}
\STATE $G_1 := G \setminus W_1;$\label{G1}
\STATE $N_h := \{v\in G_1 | \pi(v)=h\};$
\STATE $G_2 := G_1 \setminus \atr_E(N_h,G_1)$\label{G2}
\STATE $W' := \solve_O(G_2,h-1,p_O,p_E);$\label{W'}
\STATE $W_2:=\atr_O(W',G_1)$
\STATE $G_3 := G_1 \setminus W_2;$\label{G3}
\STATE $W_3:=G_3 \setminus \solve_E(G_3,h,p_E,\lfloor p_O/2\rfloor);$\label{W''}
\STATE $G := G \setminus (W_1+W_2+W_3);$
\RETURN $G$
\end{algorithmic}
\end{algorithm}

\section{Correctness}

We prove the following lemma, which guarantees that $\solve_E(G,h,p_E,p_O)$ and $\solve_O(G,h,p_E,p_O)$ partition a parity game $G$ of maximal priority at most $h$ into a region that contains all odd dominions up to size up to $p_O$ and a region that contains all even dominions of size up to $p_E$. Then $\solve_E(G,h,|G|,|G|)$ solves $G$ of maximal priority $h$.

\begin{lemma}
 $\solve_E(G,h,p_E,p_O)$, where $h$ is even and no smaller than the maximal priority in $G$, returns a set that:
 \begin{itemize}
 \item[i)] contains all even dominions up to size $p_E$, and
 \item[ii)]\label{item:no-odd} does not intersect with an odd dominion with size up to $p_O$.
 \end{itemize}
 Similarly,  $\solve_O(G,h,p_O,p_E)$, where $h$ is odd and no smaller than the maximal priority in $G$, returns a set that:
 \begin{itemize}
 \item[i)] contains all odd dominions up to size $p_O$, and
 \item[ii)]\label{item:no-even} does not intersect with an even dominion with size up to $p_E$.
 \end{itemize}
\end{lemma}

\begin{proof}
We show this by induction over the sum $h+p_E+p_O$.

\begin{description}
\item[Base case $h+p_E+p_O=0$.] Then $p_E=p_O=0$ and any set will do.

\item[Induction step] We consider the case of $\solve_E$; the case of $\solve_O$ is similar. If $h=0$, then $G$ is a dominion for Even; we are done. We proceed with $h>0$. 

We first show i) that $\solve_E(G,h,p_E,p_O)$ returns all even dominions up to size $p_E$. Let $D$ be such a dominion.
According to the IH, $D$ does not intersect with  $W$ 
and therefore it does not intersect with $W_1$ 
either. It is therefore contained in $G_1$.
The intersection $D'$ of $D$ and $G_2$ is an even dominion in $G_2$ (Lemma \ref{lem1}) and therefore, from the IH, it does not intersect with $W'$ 
nor  with its Odd attractor $W_2$ in $G_1$ (Lemma \ref{lem:att}). Then, $D$ does not intersect with $W_2$ either.
$D'$ is also contained in $G_3$ 
and by IH does not intersect with $W_3$ 
and therefore neither does $D$. Since  $D$ does not intersect with $W_1,W_2$ nor $W_3$ it is contained in the returned $G$.

We proceed with showing ii) that $\solve_E(G,h,p_E,p_O)$ returns a set that does not intersect with odd dominions of size up to $p_O$. Let $D$ be such a dominion, let $S$ be the union of odd dominions up to size $\lfloor p_O/2 \rfloor$ contained in $D$ and let $A$ be its Odd-attractor in $D$.

$S$ is contained in $W$ 
by IH, and therefore $A$ is contained in $W_1$ 
and does not intersect with $G_1$. 
If $A=D$ then $D$ is contained in $W_1$ and we are done. 

 We consider the case of $A\neq D$. $D\setminus A$ is non-empty and a dominion in $G\setminus A$ (Lemma \ref{lem1}). It contains an odd dominion $C$ of $G\setminus A$ in which $h$ does not occur (Lemma \ref{lem:no-h}).  Observe that since $C$ is an odd dominion in $G\setminus A$ and $A$ is an odd dominion in $G$, $C\cup A$ is an odd dominion in $G$. Since it is not included in $S$, it is larger than $p_O/2$. We now show that $C\cup A$ is included in $W_1+W_2$.

Since $W_1$ is an odd attractor,
$G_1\cap C$ is an odd dominion in $G_1$, and since $C$ contains no $h$, also in $G_2$.
By IH, it is contained in $W_2$ 
and $C$ is therefore contained in $W_1\cup W_2$, 
as is $A\cup C$. 

Then, since $A\cup C$ is larger than $p_O/2$, $D\setminus (W_1\cup W_2)$ is not only a dominion of $G_3$ (Lemma \ref{lem1}), it is also 
of size up to $\lfloor p_O/2 \rfloor$ and by IH contained in $W_3$. 
Hence $D$ is does not intersect with the returned $G$.
\end{description}
\end{proof}

\section{Analysis}

Let $f(h,l)$ be the number of calls to $\solve_E$ and $\solve_O$ of $\solve_E(G,h,p_E,p_O)$ (or of $\solve_O(G,h,p_E,p_O)$, if it is greater) where $l=\lfloor \log (p_E) \rfloor + \lfloor \log (p_O) \rfloor$.

An induction on $l+h$ shows that $f(h,l)\leq 2^l {h+l\choose  l} $. If $h+l=0$ then $h=0$ so $\solve_E(G,h,p_E,p_O)$ and $\solve_O(G,h,p_E,p_O)$ return immediately. For $h+l\geq 1$, we have:

\begin{equation} \label{eq1}
\begin{split}
f(h,l) & \leq 2f(h,l-1) + f(h-1,l) \\
 & \leq 2^{l-1} {{h+l-1}\choose{l-1}} + 2^{l}{{h+l-1} \choose l} \\
  & \leq 2^{l} {h+l \choose {l}}
\end{split}
\end{equation}

Then, as $l=2\dot \lfloor \log n  \rfloor $, this bring the complexity of the simplified algorithm down by a quasi-polynomial factor from Parys' version.

\begin{remark}
A $(n,d)$-universal tree is a tree into which all trees of height $d$ with $n$ leaves can be embedded while preserving the ordering of children. These combinatorial objects have emerged as a unifying thread among quasi-polynomial solutions to parity games and have therefore been the object of a recent spree of attention~\cite{CDFJRP19,FG18,CF19}. In particular, the size of a universal trees is at least quasi-polynomial, making this a potentially promising direction for lower bounds.
We observe that the call tree where the node $\solve_E(G,h,p_E,p_O)$ has for children its calls to $\solve_E$ and $\solve_O$ with parameter $h-1$  takes the shape of a universal $(n,d)$-tree where $n$ is the size of the parity game and $d$ its maximal colour. The recursive approach therefore does not seem to be free from universal trees either.
\end{remark}

\section{Conclusion}

This improvement brings the complexity of solving parity games recursively down to almost match the complexity to the algorithms based on Calude et al.'s method~\cite{CJKLS17,FJSSW17} and Jurdzi\'nski and Lazi\'c's algorithm \cite{JL17}. In particular it is fixed-parameter tractable, and polynomial when the number of colours is logarithmic. 
However, since the recursion solves the game symmetrically---that is, it goes through every colour, rather than just every other colour---and since the size of only the guarantees for the even or odd dominions are halved,
in the ${{a}\choose{b}}$ notation both $a$ ($c$ vs.\ $c/2$) and $b$ ($2\log n$ vs.\ $\log n$) double compared to Jurdzi\'nski and Lazi\'c's algorithm \cite{JL17}.

Whether this simplification to the recursion scheme makes this algorithm usable in practice remains to be seen.

\section*{Acknowledgements}
We thank  Nathana\"el Fijalkow for enlightening discussions and Pawe{\l} Parys for his comments.

\bibliographystyle{alpha}
\bibliography{refs}

\newcommand{\etalchar}[1]{$^{#1}$}
\begin{thebibliography}{CDF{\etalchar{+}}19}

\bibitem[CDF{\etalchar{+}}19]{CDFJRP19}
Wojciech Czerwi{\'n}ski, Laure Daviaud, Nathana{\"e}l Fijalkow, Marcin
  Jurdzi{\'n}ski, Ranko Lazi{\'c}, and Pawe{\l} Parys.
\newblock Universal trees grow inside separating automata: Quasi-polynomial
  lower bounds for parity games.
\newblock In {\em Proceedings of the Thirtieth Annual ACM-SIAM Symposium on
  Discrete Algorithms}, pages 2333--2349. SIAM, 2019.

\bibitem[CF]{CF19}
Thomas Colcombet and Nathana{\"e}l Fijalkow.
\newblock Universal graphs and good for small games automata: New tools for
  infinite duration games.
\newblock to appear in proceedings of FOSSACS 2019.

\bibitem[CJK{\etalchar{+}}17]{CJKLS17}
Cristian~S. Calude, Sanjay Jain, Bakhadyr Khoussainov, Wei Li, and Frank
  Stephan.
\newblock Deciding parity games in quasipolynomial time.
\newblock In {\em Proceedings of the 49th Annual ACM SIGACT Symposium on Theory
  of Computing}, pages 252--263. ACM, 2017.

\bibitem[FGO18]{FG18}
Nathana{\"e}l Fijalkow, Pawe{\l} Gawrychowski, and Pierre Ohlmann.
\newblock The complexity of mean payoff games using universal graphs.
\newblock {\em arXiv preprint arXiv:1812.07072}, 2018.

\bibitem[FJS{\etalchar{+}}17]{FJSSW17}
John Fearnley, Sanjay Jain, Sven Schewe, Frank Stephan, and Dominik Wojtczak.
\newblock An ordered approach to solving parity games in quasi polynomial time
  and quasi linear space.
\newblock In {\em Proceedings of the 24th ACM SIGSOFT International SPIN
  Symposium on Model Checking of Software}, pages 112--121. ACM, 2017.

\bibitem[JL17]{JL17}
Marcin Jurdzi{\'n}ski and Ranko Lazic.
\newblock Succinct progress measures for solving parity games.
\newblock In {\em 2017 32nd Annual ACM/IEEE Symposium on Logic in Computer
  Science (LICS)}, volume~00, pages 1--9, June 2017.

\bibitem[Leh18]{Leh18}
Karoliina Lehtinen.
\newblock A modal $\mu$ perspective on solving parity games in quasi-polynomial
  time.
\newblock In {\em Proceedings of the 33rd Annual ACM/IEEE Symposium on Logic in
  Computer Science}, pages 639--648. ACM, 2018.

\bibitem[Mar75]{Mar75}
Donald~A Martin.
\newblock Borel determinacy.
\newblock {\em Annals of Mathematics}, pages 363--371, 1975.

\bibitem[Par19]{Par19}
Pawe{\l} Parys.
\newblock Parity games: Zielonka's algorithm in quasi-polynomial time.
\newblock Available online: \url{https://arxiv.org/abs/1904.12446}, 2019.

\bibitem[Zie98]{Zie98}
Wieslaw Zielonka.
\newblock Infinite games on finitely coloured graphs with applications to
  automata on infinite trees.
\newblock {\em Theoretical Computer Science}, 200(1):135 -- 183, 1998.

\end{thebibliography}

\end{document}